\documentclass[runningheads]{llncs}
\usepackage[dvipsnames,svgnames,table]{xcolor}
\usepackage[pdftex]{graphicx}
\usepackage{amsfonts}
\usepackage{verbatim}

\newcommand{\Zed}{\mathbb Z}
\newcommand{\Nat}{\mathbb N}

\newcommand{\Real}{\mathbb R}

\begin{document}

\title{On Descriptional   Complexity of
the Planarity Problem for 
Gauss Words}

\author{Vitaliy Kurlin \inst{1}
\and Alexei Lisitsa \inst{2}
\and Igor Potapov \inst{2}
\and Rafiq Saleh \inst{2}}

\institute{
Department of Mathematical Sciences,
Durham University,  \\
Durham DH1 3LE, UK \\
\email{vitaliy.kurlin@durham.ac.uk}
\and 
Department of Computer Science,
University of Liverpool, Ashton Building, \\
Ashton St, Liverpool L69 3BX, U.K.\\
\email{Lisitsa,Potapov,R.A.Saleh@liverpool.ac.uk}
}

\maketitle

\begin{abstract}
In this paper we investigate the descriptional  complexity
of knot theoretic problems and show upper bounds for planarity problem
of signed and unsigned knot diagrams represented by Gauss words.
Since a topological equivalence of knots can involve knot diagrams with arbitrarily many crossings
 then Gauss words will be considered as strings over an infinite (unbounded) alphabet.
For establishing the upper bounds on recognition of knot properties,
we study these problems in a context of automata models over an infinite alphabet.
\end{abstract}


\section{Introduction}
\vspace*{-2mm}
Algorithmic and computational topology is a new growing
branch of modern topology. Much of the recent effort has
focused on classifying the inherent complexity of topological problems.
In this paper we investigate the  complexity
of knot theoretic problems and show upper  bounds for planarity problem
of signed and unsigned knot diagrams.
The main goal of proposed approach is to give a new insight on knot problems and
characterise knot problems according to their computational complexity.
The results presented in this paper were achieved by a combination
of methods from knot theory and automata theory. 

Knot theory is the area of topology that studies mathematical knots and links.
A knot (a link)  is a smooth embedding of a circle (several circles, respectively) in 
 3-dimensional space ${\Real}^3$, considered up to a smooth deformation of the ambient space.
It is well established and exciting area of research with strong connections 
 with topology, algebra and combinatorics. Examples of interactions between knot
theory and computer science include works on formal language theory \cite{kari1993mig},
 quantum computing \cite{abramsky2007tla,freivalds2005ktj,lomonacojr2006tqc} and
computational complexity  \cite{hass1999cck}.

Knots can be described in many ways, including various discrete representations. 
A common method of representing a knot is  a knot diagram that is a projection of 
 the knot to a plane in a general position involving only double crossings.
At each crossing we indicate which branch is "over" and which is "under",
 which allows us to recreate the original knot,  see the left picture of Figure~1 in subsection~3.1.
For oriented knots and links, each crossing has a well-defined sign (`$+$' or `$-$') 
 as shown on the right picture of Figure~1.

A knot diagram can be encoded by a string of symbols 
 $O_i$'s (over) and $U_i$'s (under) known as a {\it Gauss word}. 
The procedure of writing a Gauss word can be described as follows.
Starting from a base point on a knot diagram, write down the labels of the crossings 
 and their types of strand ordered according to the orientation of the knot.
The oriented trefoil $K$ in the middle picture of Figure~1 is encoded by 
 the \emph{signed} Gauss word ${U_1}^+ {O_2}^+  {U_3}^+ {O_1}^+ {U_2}^+ {O_3}^+$.
Removing signs leads to the \emph{unsigned} Gauss word representing the non-oriented trefoil.
Indices of crossings can be arbitrarily permuted, so one knot diagram can be encoded by many Gauss words.
A link can be represented by \emph{several} Gauss words -- one for each component of the link.

The construction of a Gauss word is  quite straightforward 
 by reading visited crossings travelling along a knot diagram. 
The inverse problem of reconstructing a knot from a strings of 
 symbols $O_i$'s and $U_i$'s is harder and is not always solvable. 
An arbitrary Gauss word may not encode a classical (planar) diagram.
In such a case the Gauss word corresponds to a non-classical knot 
 diagram containing virtual crossings that have no overcrossing information
 and are not listed in the Gauss word, see Figure~2 in subsection~3.2.
Such a crossing naturally appears after projecting a knot in a thickened surface to the plane, 
 when two branches of the knot on different sides of the surface meet in a planar projection, 
 but not in the canonical projection to the surface, see \cite[Figure~2]{kurlin2006gpr}.
So virtual knots \cite{kauffman1998vkt} are motivated by studying arbitrary Gauss words
 and by extending the classical knot theory in ${\Real}^3$ to thickened surfaces.
Most of the problems of recognising knots properties 
 (such as virtuality, unknottedness, equivalence) are known to be decidable,
with different time complexity. However their complexity in terms of computational power of devices
needed to recognise the knot properties was not studied yet. In this paper we address this problem
and provide first known bounds for some knot problems in this context.

The central problem that we study in this paper is to determine whether 
 a given Gauss word is \emph{planar}, i.e. encodes a plane diagram of a classical knot in ${\Real}^3$.
Any knot diagram can be made more complicated by introducing new trivial crossings, but keeping 
 the topological type of the knot, e.g. see Reidemeister moves on knot diagrams in \cite[Figure~4]{kurlin2006gpr}.
Hence Gauss words of knot diagrams should be considered as strings over an infinite (unbounded) alphabet.
In this context we cannot estimate computational complexity in terms of classical
models over finite alphabets and need to consider a new hierarchy of languages and models over
an infinite alphabet. Such models were recently introduced in \cite{kaminski1990fma,neven2004fsm}.

In Section~\ref{Models} we describe and extend the models of automata over an infinite alphabet
 that we use for establishing the lower and upper bounds on recognition of knot properties.
Then in Section~\ref{Results} we show that the language of planar (non-planar) signed Gauss words 
 can be recognised by deterministic two-way register automata simulating  
 the recently discovered linear time algorithm in \cite{kurlin2006gpr}. 
Due to the fact that the algorithm presented in \cite{kurlin2006gpr} allows us 
 to check planarity property not only for knots but also for links, 
 we think that the proposed idea of recognising planarity by register automata 
 can be extended to links after minor modification. 
The result is final in the sense that the power of
non-deterministic one-way register automata is not even enough 
 to recognise whether an input is a Gauss word.
In \cite{lisitsa2009automata} we conjectured  that the planarity problem for unsigned Gauss words is harder 
 than the the same problem for signed Gauss words and cannot be solved by non-deterministic register or 
 $k$-pebble automata over an infinite alphabet. Here we show that co-nondeterministic register automata are 
capable to recognize unsigned planar Gauss words.  
For the case of unsigned Gauss words we provide an upper bound
by showing that planarity can be checked by deterministic linearly bounded memory automata.

\section{Automata over Infinite Alphabets}\label{Models}

Let $D$ be an infinite set called an \emph{alphabet}. A word, or a string over $D$, or shortly,
$D$-word or $D$-string, is a finite sequence 
 $d_{1}, \ldots ,d_{n}$, where $d_{i} \in D$, $i=1, \ldots, n$.
A language over $D$ (a $D$-language) is a set of $D$-words.
For a word $w$ and a symbol $d$, denote by $\mid w\mid_d$ the number of occurrences of $d$ in $w$.
As usual $\mid w \mid$ denotes the length of the word $w$. A language $L$ over an infinite alphabet $D$ is called \emph{$n$-bounded} if there is a constant $n \in \Nat$ such
that for any $w \in L$ and for any $d \in D$ one has $\mid w\mid_d \le n$.  
All languages we consider are bounded.

\subsubsection{Register Automata}

Register automata are finite state machines equipped with a finite number of memory cells called registers which may hold values from an infinite alphabet. It is one of the weakest models of automata over infinite alphabets introduced in \cite{kaminski1990fma} and studied further in \cite{neven2004fsm}.

\begin{definition}[\cite{neven2004fsm}]
A non-deterministic two-way k-register automaton over an infinite alphabet D is a tuple (Q,q{\tiny 0}{\small ,F,}{\large $\tau$}{\tiny 0}{
,P) where} {Q is a finite set of states, q{\tiny 0} $\in$ Q is the initial state,F $\subseteq$ Q is the set of final states, }{\large $\tau$} {\tiny 0}{
: \{1,...,k\}$\rightarrow$ D $\cup$ \{$\vartriangleright$,$\vartriangleleft$ \}is the initial register assignment and P is a finite set of transitions of the forms:}
\end{definition}
\begin{itemize}
\item[$1)$] $(i,q) \rightarrow (q',d)$
(If a current state is $q$ and the observed symbol on the tape equals to a value in the register $i$
then enter the state $q'$ and move along the string according to the specified direction $d$;)
\item[$2)$] $q \rightarrow (q',i,d)$
(If a current state is $q$ and the observed symbol on the tape does not equal to any value held in registers then enter the state $q'$, copy the current symbol to a specified register $i$ and move along the string according to the specified direction $d$,
where $i \in \{1,..,k\}, q,q' \in Q$ and $d \in \{stay, left, right\}.$)
\end{itemize}

Given a $D$-word $d$ delimited by symbols $\vartriangleright$,$\vartriangleleft$ on the input tape, an automaton starts in a state $q_{0}$ and in the position of the first letter of $d$ and applies non-deterministically any applicable rules. As usual, if automaton is able ever to reach a state $q \in F$, it accepts the word, otherwise the word is rejected. The set of all accepted words forms a language recognisable by an automaton.
An automaton is \emph{deterministic} if in each configuration at most one transition applies.

For the purpose of this paper we modify the definition of register automata from \cite{neven2004fsm} by allowing more general transition rules that allows replication of the same value in different registers. This does not affect the computational power of the model (see Lemma~\ref{lemma:modified} below), but makes the design of such automata for various recognition problems much more natural and easier. Similar modifications (in more general setting) have appeared in \cite{david2004mdi,demri2006lfq}.

We define  \emph{modified} two-way k-register automata by adding to the definition above two extra types of transition rules:
\begin{itemize}
\item[$3)$] $(i,q) \rightarrow (q',j,d)$
(If  a current state is $q$ and the observed symbol equals to a value in the register $i$
then enter the state $q'$, copy the current symbol to a register $j$
and move along the string according to the specified direction $d$);
\item[$4)$] $q \rightarrow (q',d)$
(If a current state is $q$ and the observed symbol does not equal to any value held in registers then enter the state $q'$
and move along the string according to the specified direction $d$,
where $i,j \in \{1,..,k\}, q,q' \in Q$ and $d \in \{stay, left, right\}$).
\end{itemize}

\begin{lemma}~\label{lemma:modified}
The models of original register automaton and modified register automaton over an infinite alphabet are equivalent.
\end{lemma}

\begin{proof}
We show that two extra types of rules of the modified model can be simulated by the original automata.

The rule of type 4 is simulated by adding one extra dummy register $k+1$ and replacement of rules of modified
automata of the form $q \rightarrow (q',d)$  by a pair of rules $(k+1,q) \rightarrow (q'd)$ and $q \rightarrow (q',k+1,d)$.
If a value of the register $k+1$ is equal to the observed symbol then the rule $(k+1,q) \rightarrow (q'd)$ is applicable otherwise the rule  $q \rightarrow (q',k+1,d)$ is applicable.
Also we replace the rule $q \rightarrow (q',i,d)$ of type 2  in the modified automata model by a pair of rules; original type 2 rule
$q \rightarrow (q',i,d)$ and type 3 rule  $(k+1,q) \rightarrow (q'd)$.

The rule of type 3 $(i,q) \rightarrow (q',j,d)$ of the modified  model
that allows storing the same value in different registers, can be simulated in original model
by using the following construction.

The state of the registers of the modified automaton, that is  a sequence of not necessarily different values
$R=[r_1,r_2,\ldots ,r_{k+1}]$ is represented in the simulating automaton as a pair:
\begin{itemize}
\item the set of unique values $U=\{r_1,r_2,\ldots ,r_{k+1}\}$, and
\item the surjective mapping $\phi: \{1,\ldots,k+1\} \rightarrow U$
\end{itemize}
The content of $U$ is kept in the registers and since the mapping
$\phi$ is finite, it can be kept in the finite state control. Now
it is straightforward to simulate the effects of all  possible
types of rules, including the type 3, in terms of pairs $U,\phi$.
We omit obvious details.\qed

\end{proof}

\subsubsection{Pebble Automata}
As an alternative model of automata over infinite alphabet,  \emph{pebble automata}(PA) 
 were introduced in \cite{milo2003txt} and further studied in \cite{neven2004fsm}. We follow the
definitions in \cite{neven2004fsm}.  In this model, instead of registers, finite state machines are equipped with the finite set of pebbles which can be placed on the input string and later lifted following the \emph{stack} discipline. That means pebbles are numbered from $1$ to $k$ and pebble $i+1$ can only be placed when pebble $i$ has already been placed on the string and vice-versa, pebble $i$ can only be lifted if $i+1$ is not on the string.
Further assumption is that the pebble with the highest number acts as a head, so an automaton has an access to the symbol of the string under such a pebble and to the information on which other pebbles are located at the same position.
The transition of pebble automata depends on the following: the current state, the pebbles placed on the current position of the head, the pebbles that see the same symbol as the top pebble. The effect of the transition is the change of a state, movement of the head and, possibly, removal of the head pebble, or placement of the new pebble.
As usual acceptance of a word is defined as reachability of one of the final states.

As expressive power concerned, in general pebble automata are incomparable with register automata \cite{neven2004fsm}. We will show, however, in the Section~\ref{Results} that over a class of \emph{bounded languages}, including all languages of our interest,
PA can be effectively simulated by RA.

\subsubsection{Linearly Bounded Memory Automata}
In all models above the input can be thought of as  given on the input tape which can only be read,
but not written on. Linearly bounded memory automata (LBMA) are an extension of register automata with the input tape.
The automaton can read and write in the tape cells symbols of an infinite alphabet.
The input is given on the initial part of the tape and for the input size $n$, the size of the tape is assumed to be $O(n)$, i.e. linearly bounded. Types of rules of LBMA include all types of rules of (modified) RA and additional rules allowing us to write on the tape. For every form $L \rightarrow (\ldots)$ of rules of the (modified) RA model the following is a form of rule for LBA: $L \rightarrow (\ldots, i)$, where
$i \in \{1,\ldots k\}$. The effect of application of the latter is the same as of the former,  plus the automaton \emph{writes the content of the register $i$} in the current position on the tape before possible head movement.

\subsubsection{Words and Data Words}
In previous works on the computational models on infinite alphabets it has been acknowledged that in many situations it is natural to consider infinite alphabets as the subsets of  $\Sigma \times \Delta$ where $\Sigma$ is a \emph{finite}  set  and $\Delta$ is an infinite set. Thus, the symbol here is an ordered pair $(a,b)$. The words over such alphabets are called \emph{data words}
\cite{bjorklund2007nrd}. In the definition of automata over data words, it is sensible to assume that when an automaton reads a symbol $(a,b)$ it has a direct access to both components of the pair.  For this purpose, the form of transition rules can be adapted to include one extra argument on the left-hand sides.
For example, the rule  $(c,i,q) \rightarrow (q',d)$ is read as "if an automaton is in a state $q$ and  observes the symbol $(c,a)$ and $a$ is the content of the register $i$ then the automaton  can change the state to $q'$ and move the head along the direction $d$." It should be clear now how to modify the definitions of all above models to work over data words.
\section{Recognisability of Knot Properties}\label{Results}
\subsection{The Language of Gauss Words}
Knots defined as smooth embeddings of a circle in 3-dimensional Euclidean space 
 can be combinatorially encoded by finite structures, such as graphs or words. 
One of such discrete representations is a Gauss word consisting of a sequence of symbols 
(labels O ("over") and  U ("under") with indices and signs), which
 can be read of a projection of the knot to the plane.
Fixing an orientation of a knot, one can distinguish between positive and
 negative crossings labelled in Figure~1.

\begin{center}
\includegraphics[scale=0.35]{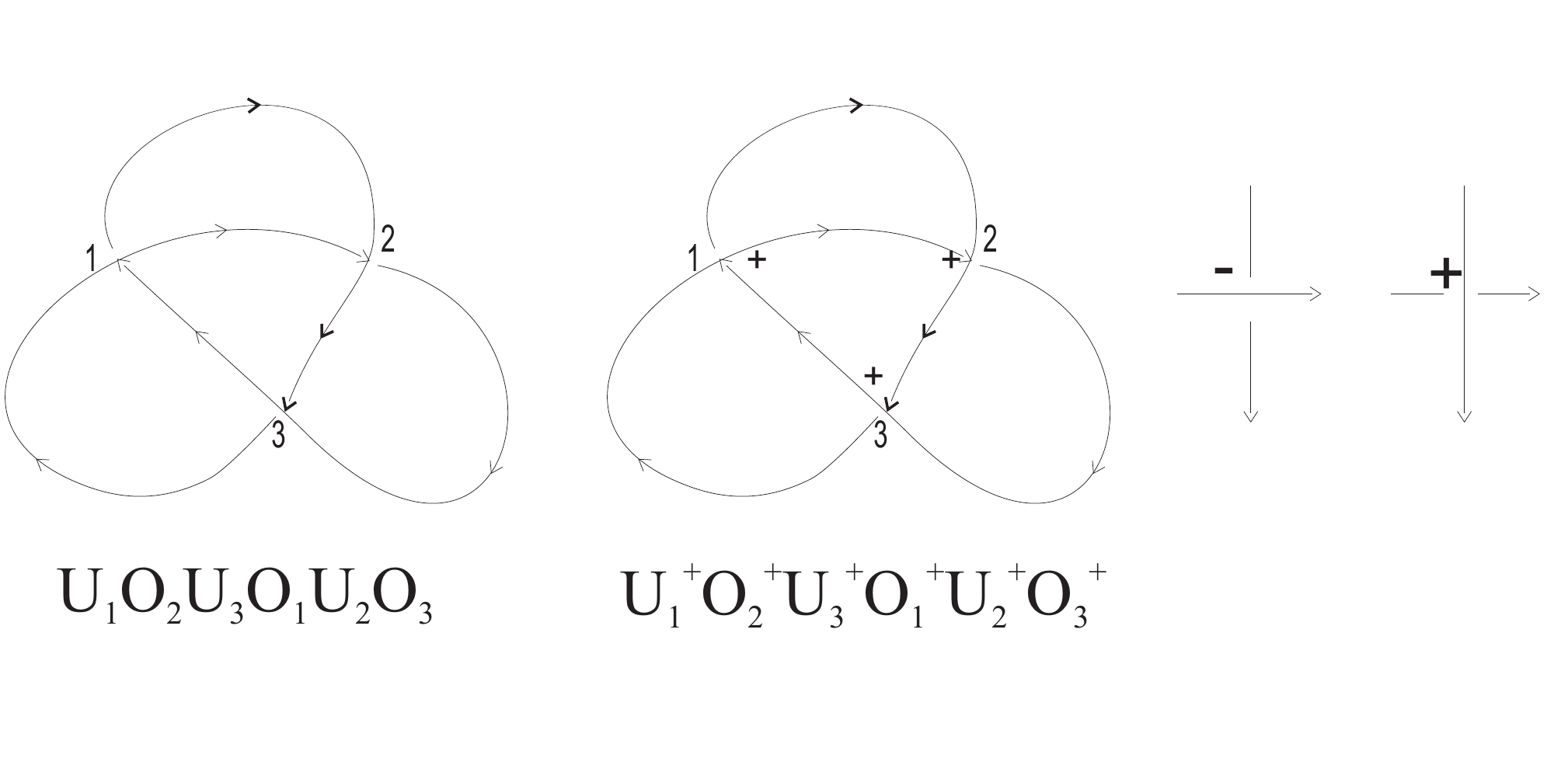}\\
\textbf{Fig. 1.} A knot diagram of a trefoil with its unsigned and signed Gauss words
\end{center}

\begin{definition}
An unsigned Gauss word $w$ is a data word  over the alphabet $\Sigma \times \Nat$ where $\Sigma=\{U,O\}$, such that for every $n \in \Nat$ either
\begin{itemize}
\item $|w|_{(U,n)} = |w|_{(O,n)} = 0$, or
\item $|w|_{(U,n)} = |w|_{(O,n)} = 1$
\end{itemize}
\end{definition}

The formal conditions above are geometrically justified by non-oriented knot diagrams,
 where the $n$-th crossing is labelled by $U_n$ (`under') and $O_n$ (`over').
In signed Gauss words of oriented knot diagrams each crossing 
 also appears twice and with the same sign, which is formalised below.

\begin{definition}
A signed Gauss word $w$ is a data word over the alphabet $\Sigma \times \Nat$ where  $\Sigma=\{U^+,O^+,U^-,O^-\}$,  such that for every $n$ either
\begin{itemize}
\item $|w|_{(U^{+},n)} = |w|_{(O^{+},n)} = |w|_{(U^{-},n)} = |w|_{(O^{-},n)} = 0$, or
\item $|w|_{(U^{+},n)} = |w|_{(O^{+},n)} = 1$ and $|w|_{(U^{-},n)} = |w|_{(O^{-},n)} = 0$, or
\item $|w|_{(U^{-},n)} = |w|_{(O^{-},n)} = 1$ and $|w|_{(U^{+},n)} = |w|_{(O^{+},n)} = 0$,
\end{itemize}
\end{definition}

\begin{proposition}
The language $L_{GW}$ of unsigned Gauss words and the language $L_{SGW}$ of signed Gauss words 
 are recognisable by deterministic 2-way register automata.
\end{proposition}
\begin{proof}
We explain only the construction of a 2-way deterministic register automaton $A$ which recognises $L_{GW}$. With obvious modifications the automaton can be adapted to the case of $L_{SGW}$.
Let $w$ be a data word $(a_1,b_1)...(a_n,b_n)$ such that $a \in \Sigma =\{U,O\}$ and $b \in \Nat$.
The automaton $A$ reads the first symbol  $(a_i,b_i)$ and stores the value of $b_i$ in some register, then it moves right then left along the word to compare the current symbol $(a_j,b_j)$  with the value of $b_i$ held in some register. If the symbol $(a_j,b_j)$  where $b_i = b_j$ and $b_i \ne a_j$ is found and there are no further occurrences of $b_i$, then the automaton $A$ moves right along the word and checks the next symbol. 
If the next symbol is equal to the end symbol then $A$ moves to an accepting state.
\end{proof}

\begin{proposition}
The languages $L_{GW}$  and $L_{SGW}$ are not recognisable  by non-deterministic
one-way register automata.
\end{proposition}
\begin{proof}
We show the argument only for the case of $L_{GW}$. With obvious modifications it works for
$L_{SGW}$  as well. The argument is not new and was used e.g. in \cite{bjorklund2007nrd} 
 to show non-recognisability of some data languages by one-way register automata.
Assume that language $L$ is recognisable by some one-way register automaton $A$ with $n$ registers.
Consider the word 
$$w = (U,1)(U,2) \ldots (U,n+1)(O,1)(O,2) \ldots (O,n+1) \in L.$$ 
The automaton $A$ accepts this word.
After reading first $n+1$ positions, there is at least one index value $i\in \{1, \ldots, n+1\}$  
 which does not appear in any register of $R$. 
That means that automaton $A$ also accepts a word $w' \not\in L$ 
 obtained from $w$ by replacing $(U,i)$ with $(U,i+1)$. 
This contradicts the assumption on $A$.
\end{proof}

\subsection{Planar  and Non-planar Gauss Words}
Every knot can be represented by a Gauss word, but not every Gauss
word represents a classical knot in ${\Real}^3$. 
 For example, any attempt to reconstruct 
 a knot diagram from the Gauss word $O^{-}_{1} O^{-}_{2} U^{-}_{1}
 O^{+}_{3} U^{-}_{2} U^{+}_3$ inevitably leads to new (\emph{virtual})
crossings which are not present in the Gauss word, see Figure 2.
This observation was one of the motivations for introducing
\emph{virtual knot theory} \cite{kauffman1998vkt}.  
A Gauss word representing a knot diagram on a plane without virtual crossings 
 is called \emph{classical} or \emph{planar}. The problem of recognising planar
 Gauss words have been formulated by Gauss himself and recently
 several algorithmic solutions for both signed and unsigned cases
 have been proposed, e.g. in
\cite{cairns1993pps,cairns1996ppi,kauffman1998vkt,kurlin2006gpr,read1976gauss,lovasz1976forbidden,manturov2005proof,rosenstiehl1984gauss,de1999characterization,shtylla2009realization}.

In this section we address the question of recognising planarity of Gauss words
 by automata models considering the unsigned and signed cases separately.

\begin{center}
\includegraphics[scale=0.30]{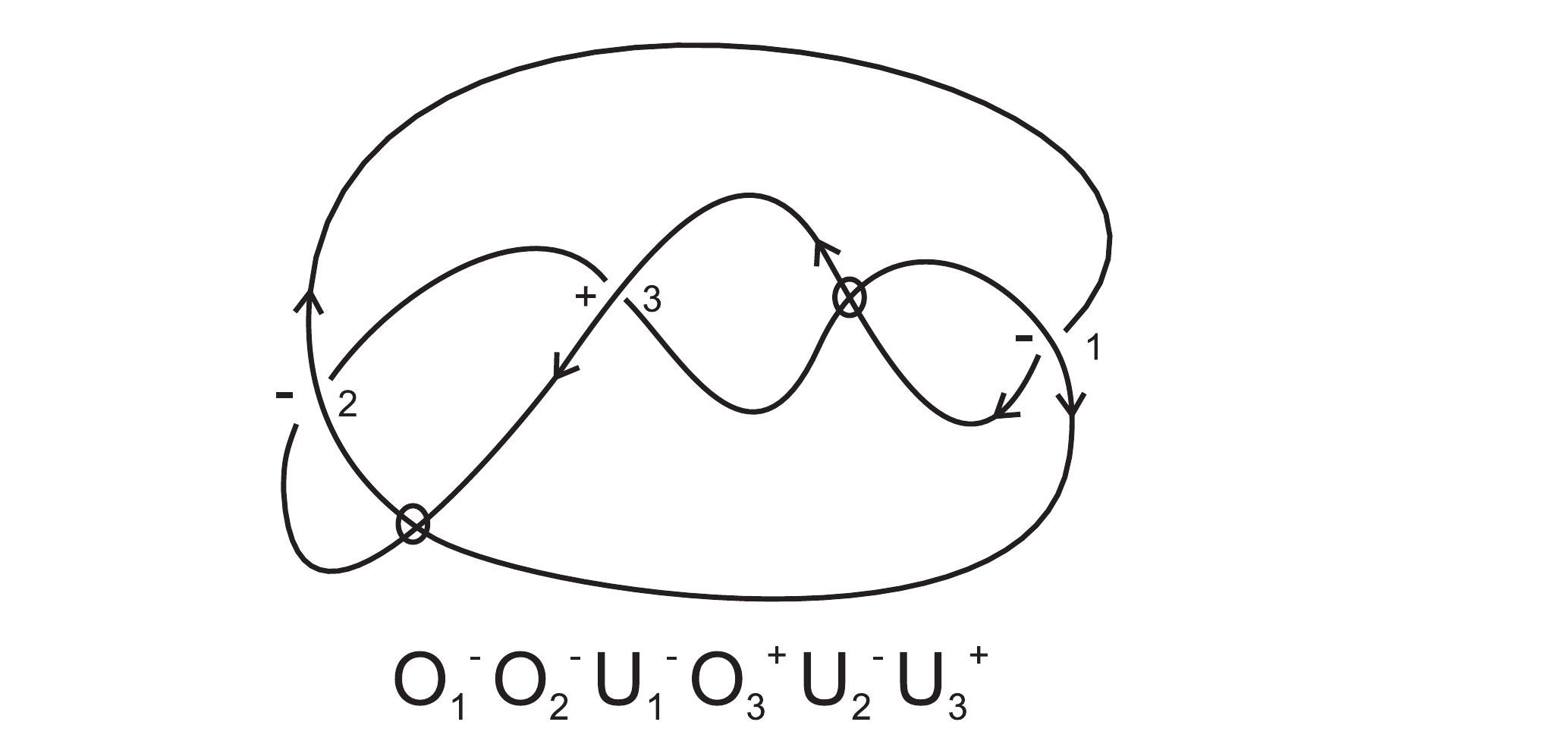}\\
\textbf{Fig. 2.} An example of a virtual knot with 2 virtual crossings
\end{center}

\subsubsection{Signed Gauss Words}

First we will show that the language of planar signed Gauss words 
 can be recognised by two-way deterministic register automata.
The design of automata will be based on the linear time algorithm 
 presented by V.~Kurlin in \cite{kurlin2006gpr}.
The main idea of the algorithm is to find the least genus of the surface 
 containing a knot diagram without virtual crossings encoded by a given Gauss word. 
For this purpose the Euler characteristic $\chi$ of 
 the \emph{combinatorial Carter surface} \cite{carter1991cic} 
 associated with the knot diagram is computed as the number of faces (cycles) 
 minus the number of edges plus the number of vertices.
In the context of Gauss words, the number of edges is the length of the word, 
 the number of vertices is the number of crossings (half the length of the word) 
 and the number of faces can be found by implementing traversal rules below.
\begin{lemma}~\label{T_Faces}
A two-way deterministic register automaton can traverse all faces 
 containing a crossing $i$ in the Carter surface associated with a knot diagram.
\end{lemma}
\begin{proof}

The traversal of a face containing a crossing $i$ in the knot diagram 
 can be done by choosing an initial direction and turning left 
 at each consecutive visited crossing starting from $i$.
This global property of ``turning left'' can be defined by 
 a deterministic set of traversal rules that take into
 account only local property of the current crossing and 
 some finite information about the previously visited one.
In general we have 8 cases since there are 2 types of crossings (positive and negative)
 and 4 directions from which we can approach each crossing, see Figure~3.
From a topological point of view, we take the graph defined by a Gauss word and 
 attach faces to cycles always turning left at every vertex, which leads to a surface.
If the resulting surface is a 2-dimensional sphere whose Euler characteristic is 2 then 
 the given Gauss word is planar since we have embedded the graph into the plane,
 see more details in \cite[Lemma~3.2]{kurlin2006gpr}.

\begin{center}
\includegraphics[scale=0.35]{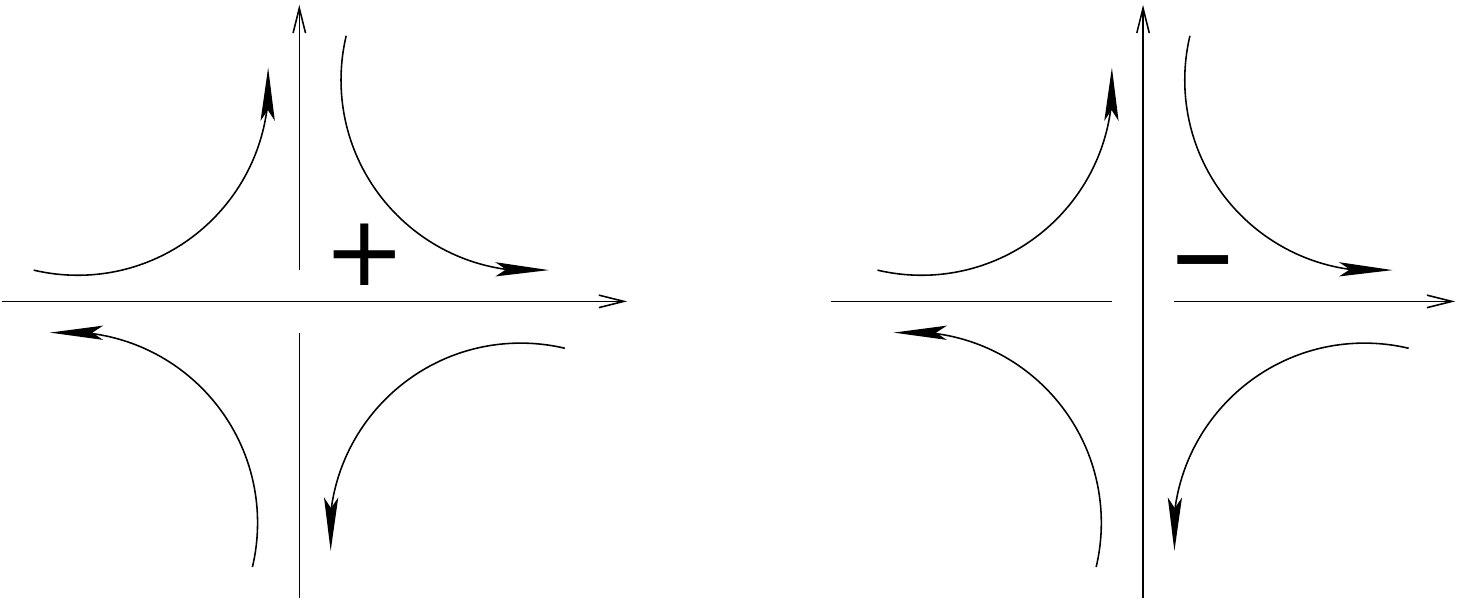}\\
\textbf{Fig. 3.} Automaton moves
\end{center}

We follow the interpretation of the local rules 
 for selecting cycles defined in \cite{kurlin2006gpr}, but
 present them here in slightly different notation, which is 
 more appropriate for the design of  a register automaton.
A register automaton that is observing a current symbol $S$, needs to choose a correct symbol
 corresponding to the next crossing after turning (geometrically) to the left on a knot diagram.
In fact on the Gauss code it will correspond to finding $S'$ that is the counterpart of $S$ 
 and then choosing a symbol which is either a left or right neighbour of $S'$. 

For example, if $S$ is $O_i$ ($U_i$, respectively) with any sign $\varepsilon=\pm$ then 
 we choose a neighbour of $U_i$ ($O_i$, respectively) with same sign $\varepsilon$ in the Gauss word. 
Geometrically, taking the right neighbour in the Gauss word is equivalent to going in the direction of the orientation of 
 the corresponding knot diagram, while taking the left neighbour means moving in the direction opposite to the orientation. 
In order to define whether we need a neighbour from the left or from the right side 
 we need to know the current type of the crossing which is $S$ and the information about 
 the previous choice of direction, i.e. whether $S$ was chosen as a left or a right symbol.
Now we define eight rules in the form  $(D,S) \rightarrow (S',D')$, where 
 $D,D'\in \{\mbox{Right},\mbox{Left}\}$ and $S,S'\in \{U,O\} \times \mathbb{N} \times \{+,-\}$.
Each rule can be read as follows: if the current symbol $S$ is reached via direction $D$ 
 then find $S'$ (counterpart of $S$) and move one step in the specified direction.

\begin{center}
\begin{tabular}{l r}
 (Right, O$_{i}$$^{+}$)$\rightarrow$ (U$_{i}$$^{+}$, Right)&
\hspace{25pt} (Right, U$_{i}$$^{+}$)$\rightarrow$ (O$_{i}$$^{+}$, Left)\\
 (Left, O$_{i}$$^{+}$)$\rightarrow$ (U$_{i}$$^{+}$, Left) &
\hspace{25pt} (Left, U$_{i}$$^{+}$)$\rightarrow$ (O$_{i}$$^{+}$, Right)\\
 (Right, U$_{i}$$^{-}$)$\rightarrow$ (O$_{i}$$^{-}$, Right) &
\hspace{25pt} (Right, O$_{i}$$^{-}$)$\rightarrow$ (U$_{i}$$^{-}$, Left) \\
 (Left, U$_{i}$$^{-}$)$\rightarrow$ (O$_{i}$$^{-}$, Left) &
\hspace{25pt} (Left, O$_{i}$$^{-}$)$\rightarrow$ (U$_{i}$$^{-}$, Right) \\
\end{tabular}
\end{center}

The first rule $(\mbox{Right}, O_{i}^{+})\to(U_{i}^{+},\mbox{Right})$ 
 is illustrated in the left picture of Figure~3, where we come to 
 the positive $i$-th overcrossing ($O_{i}^{+}$) in the horizontal direction according 
 to the orientation of the knot (Right), turn left and go out of the positive $i$-th undercrossing 
 ($U_{i}^{+}$) in the vertical direction also coinciding with the orientation (Right).
The remaining 7 rules correspond to other 7 curved arrows 
 showing `the left turn' in Figure~3.
The sign of a crossing does not change, while 
 any overcrossing is replaced by an undercrossing and vice versa.

Following the rules above, we can design a register automaton that
keeps the finite information about its previous choice of direction
(Right or Left) in its state space and chooses the Right
 or Left symbol of $S'$ after observing the symbol $S$. 
It can also keep records on which rule was applied to the starting symbol $S$ 
 and will terminate the traversal of a face if the same rule will be
 applied for $S$ again.  The fact of the repetition corresponds to
the completion of a cyclic path. In order to traverse all faces
which are adjacent to a crossing $i$, we need to start from two
different initial conditions associated with labels ($O_i$ or
$U_i$) and two different initial direction (Left or Right).\qed
\end{proof}

\begin{lemma}
Two-way register automata with $k$ registers on the input with $t$ distinct symbols 
 can simulate a counter machine with $k$ counters bounded by $t$.
\end{lemma}
\begin{proof}
Let us assume that a word on an input tape has at least $t$
distinct symbols. The value of a counter stored in a register $i$
corresponds to the number of distinct symbols from the beginning
of the word till the position of the first appearance of symbol
stored in the register. Then we can increase (decrease) the value
by looking for the next (previous) symbol on the string that will
appear for the first time. Counter $i$ is equal to zero if the
value stored in the  register $i$ is the first symbol on the input
tape. If we use $k$ registers then we can store $k$ counters
bounded by $t$, where $t$ is a number of distinct symbols on the
input tape.\qed
\end{proof}

\begin{lemma}
Two-way deterministic register automata can compute 
 the Euler characteristic of the Carter surface associated to 
 a signed Gauss word according to the rules from the proof of Lemma~2.
\end{lemma}

\begin{proof}
To compute the Euler characteristic we count the numbers of edges and vertices 
 in  the graph $G$ represented by a signed Gauss word.
Geometrically $G$ is the underlying graph of the knot diagram encoded by 
 the Gauss word and all its vertices (crossings of the diagram) have degree 4 only. 
The number of vertices in $G$ is the number of distinct symbols, 
 while the number of edges is twice as much.
Both values can be counted in a straightforward way. 
The number of faces attached to the graph in the combinatorial Carter surface 
 can be counted by traversing $G$ in the following way. 
The automaton goes sequentially through the list of vertices. For each vertex $i$ it traverses  
 (as described in Lemma \ref{T_Faces}) all adjacent faces  and increases 
 the counter by one for every face $F$ not containing vertices with indices less than $i$.
Also the automaton  counts how many times the crossing $i$ is met during 
 the traversal of faces adjacent to $i$. As soon as the value  reaches
 $4$ the automaton starts the traversal for the next crossing.
The computation of the Euler characteristic $\chi$ is done by
counting the values for edges, vertices and faces in individual
counters and then by subtracting number of edges from the sum of
the numbers of vertices and faces. Since the number of each value
in counters is bounded by the number of distinct symbols, the
computation can be done by the two-way deterministic register
automaton.\qed
\end{proof}

\begin{theorem}
The language of planar signed Gauss words can be recognised 
 by two-way deterministic register automata.
\end{theorem}

\begin{proof}
Compute the Euler characteristics by the two-way deterministic
register automaton. If the Euler characteristics $\chi$ is equal to 2, 
 i.e. the combinatorial Carter surface is a sphere, then a signed Gauss word is planar
\cite{carter1991cic,kurlin2006gpr}.\qed
\end{proof}

\subsection{Unsigned Gauss words}
As for signed words, an unsigned Gauss word $w$ is called \emph{planar} if 
 $w$ is a Gauss word read off a planar projection of a classical knot in ${\Real}^3$.
In \cite{lisitsa2009automata}, we showed that an algorithm for the decision 
of  planarity of  unsigned Gauss words , proposed by 
Kauffman in  \cite{kauffman1998vkt},  can be implemented by 
Linear Bounded Memory Automata. 
We also conjectured in \cite{lisitsa2009automata} that both deterministic and non-deterministic register automata 
 (over an infinite alphabet)  are not powerful enough  for the recognition of planarity in unsigned case.  
In this section  we show, that, surprisingly to us, the \emph{non-planarity} of unsigned Gauss word can be 
 recognised by non-deterministic register automata, and, therefore, for recognition of planarity the finite memory 
 and co-non-determinism are sufficient. Notice, that it is unknown whether the class of languages recognised by NRA is 
 closed under complementation, so the original conjecture from \cite{lisitsa2009automata} still stands.     

As the main result of this section, we show that the algorithm for non-planarity of 
 unsigned Gauss words from \cite{cairns1996ppi} is implementable in NRA.

We start with definitions which will be used to  
describe the main steps of the algorithm from  \cite{cairns1996ppi}.

\begin{definition}
For a Gauss word $w$, denote by $\alpha_i(w)$ the number of symbols that occur in $w$ 
 in cyclic order between the symbols U$_{i}$ and O$_{i}$, taken modulo 2.  
\end{definition}
 
Notice that, due to the fact that every label appears twice in a Gauss word,  
 in the above definition one can swap U$_{i}$ and O$_{i}$, so 
 the definition of $\alpha_i(w)$ will not  be affected. 

\begin{definition}
Let $S_{i}$ denote the subset of symbols that occur in $w$ in cyclic order between, 
 either the symbols $U_{i}^{+}$ and $O_{i}^{+}$, or $O_{i}^{-}$ and $U_{i}^{-}$.
Let $\bar{S}_{i}$ denote $\{U_{i}^{\pm},O_{i}^{\pm}\}\cup S_{i}$ and.
 S$_{i}$$^{-1}$ denote the set S$_{i}$ after reversing the first component of each letter, i.e. from U to O and vice versa. 
Denote by $\beta_{ij}(w)$ the number of symbols in the intersection $\bar{S}_{i}$ $\cap$ S$_{j}$$^{-1}$.
 
\end{definition}
\begin{definition}
Given an unsigned Gauss word $w$, the vertices of the interlacement graph $G(w)$ are labels in $w$ and the edges of $G(w)$ are the pairs of labels $(i,j)$ such that the labels $i$ and $j$ are interlaced ($i$ occurs once between two occurrences of $j$ and vice versa). 
\end{definition}

\begin{definition}
Given a unsigned Gauss word $w$,a signing of $w$ is a mapping from the set of labels in $w$ to the set of signs $\{+,-\}$. 
\end{definition}

\begin{example}
Given $w = U_1O_3U_4U_2O_1U_5O_2U_3O_5O_4$, the interlacement graph $G(w)$ of $w$ is shown in figure 4.  
Compute the values of $\alpha_i(w)$ and $\beta_{ij}(w)$, for $i = 1$ and $j=2$. 
We have   $\alpha_1(w) = |\{O_3U_4U_2\}|=3\equiv 1 \pmod{2}$ and 
 $\beta_{12}(w) = |\bar{S}_{1} \cap S_{2}^{-1}| = \{U_1O_3U_4U_2O_1\}\cap \{U_1O_5\} = |\{U_1\}|=1$
\begin{center}
\includegraphics[scale=0.35]{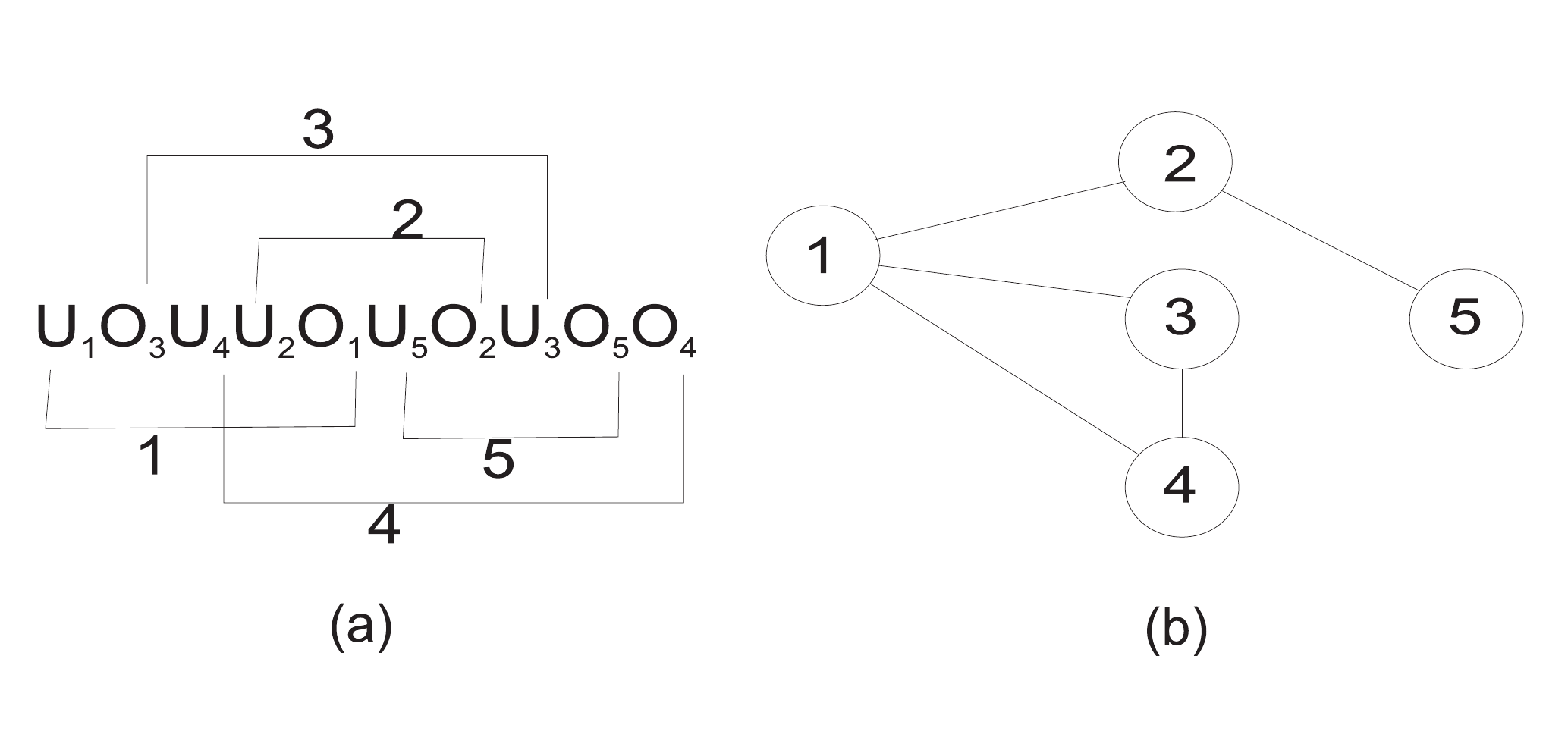}
$\newline$
\small{ \textbf{Fig. 4.} Non-planar Gauss word $w$ and its corresponding interlacement graph $G(w)$}
\end{center}
\end{example}

For each edge $e_{ij}$ in $G(w)$, we assign the number $\beta_{ij}(w)$ (Mod 2) $\in$ $\Zed_2$. 
According to \cite{cairns1996ppi} that assignment defines a $\Zed_{2}$ $1$-cochain $B(w)$ and 
 the property that the Cairns-Elton algorithm checks is whether this co-chain is closed. 
 For the purpose of this paper we need only characterisation of the closeness of $B(w)$ 
 in terms of notions we have already introduced: $B(w)$ is closed if and only if 
 for every closed path $P_i$ in $G(w)$, the sum of the numbers $\beta_{ij}(w)\equiv 0\pmod{2}$. 

Given an unsigned Gauss word $w$, the Cairns-Elton algorithm \cite{cairns1996ppi} proceeds by checking that 

\begin{enumerate}
\item $\alpha_i(w)$ is $0$ for every $i$ in $w$, 
\item $\beta_{i,j}(w)$ is even whenever $\bar{S_i}$ and $\bar{S_j}$ do not interlace and 
\item For any signing $w'$ of $w$, the co-chain $B(w)$ is closed. 
\end{enumerate} 

If conditions 1,2 and 3 are satisfied then the algorithm returns "the word $w$ is planar", 
 otherwise if any of the conditions is not satisfied, the algorithm returns "the word $w$ is non-planar".
The checking of first two conditions can be implemented by DRA. For the third one, 
 according to \cite[Lemma~1]{cairns1996ppi}, the closeness of $B(w)$ depends only on the unsigned word. 

This means we do not require to check every signing of $w$, however we are required to compute the sum of numbers assigned to the edges of every closed path in $G(w)$ and this does not seem to implementable using finite memory. Therefore to overcome this we can reformulate the conditions for the inverse property and then check with a NRA on whether there exists a closed path in $G(w)$ such that the sum of the numbers assigned to its edges is odd. 

Given an unsigned Gauss word $w$, the algorithm proceeds in three stages.

\begin{enumerate}

\item The input word is checked on whether the number of neighbours of $v_i$ is odd for any $v_i$ in $G(w)$. 
If "yes", the algorithm stops with the result "the input word is non-planar", otherwise the algorithm proceeds to the second stage.

\item The input word is checked on whether  $\beta_{ij}(w)$ is odd for any pair of vertices $(v_i,v_j)$  in $G(w)$ that are not connected by an edge.  If "yes", the algorithm stops with the result "the input word is non-planar", otherwise the algorithm proceeds to the third stage.

\item For an arbitrary signing $w'$ of $w$, the input word is checked on whether there exists a closed path in $G(w)$ such that the sum of the numbers assigned to its edges is odd.  If "yes", the algorithm stops with the result "the input word is non-planar", otherwise the algorithm stops with the result "the input word is planar".

\end{enumerate}
  
\begin{theorem}
The language of non-planar unsigned Gauss words can be 
 recognised by a two-way non-deterministic register automaton.
\end{theorem}

\begin{proof}
Let $w$ be an unsigned Gauss word and $G(w)$ be the interlacement graph of $w$. 
Denote by $N(v_{i})$ the set of all neighbours of a vertex $v_i\in G(w)$.
For the first condition, the automaton will check the number $|N(v_{i})|$ of neighbours of each vertex $v_{i}\in G(w)$.
It will store in the register the first occurrence of the label $i$ in $w$ which corresponds to vertex $v_{i}$ 
 and store the parity of the number $|N(v_{i})|$ in finite state control of the automaton. 
To check $|N(v_{i})|$, the automaton goes through each symbol $j$ in between the pair of the labels $i$ and $i^{-1}$ in $w$ (where $i^{-1}$ represents the second occurrence of $i$ in $w$) and if $j$ occurs only once then it moves first to an odd state and then alternates between odd and even states for any further occurrences. If $i^{-1}$ is reached and the current state is odd then it moves to an accepting state.
Otherwise if, for all vertices $v_{i}\in G(w)$, the number $|N(v_{i})|$ is even then it checks condition (2).

For the second condition, the automaton will check the number  $|N(v_{i})\cap N(v_{j})|$ 
  of common neighbours for any pair of vertices that are not connected by an edge of $G(w)$.  
To verify that a vertex $v_i$ is not connected to any vertex $v_j$, the automaton stores in the registers the first occurrence of $i$ and the first occurrence of each $j$ and then it checks whether there is an even number (either 2 or 0) of occurrences of each $j$ in between $i$ and $i^{-1}$. If the number of occurrences of $j$ in between $i$ and $i^{-1}$ is even then it stores the symbol $k$ in the register (which occurs in between $i$ and $i^{-1}$) and compares it with the symbols in between $j$ and $j^{-1}$. The value of $|N(v_{i})\cap N(v_{j})|$ is stored in finite state control of the automaton. If there is a match, it will move first to an odd state and then alternate between odd and even states for any further matches.  If $i^{-1 }$ is reached and current state is odd, the automaton moves to an accepting state.
Otherwise if, for all pairs of vertices $v_i,v_j$ not connected by an edge of $G(w)$,
 the number $ |N(v_{i}) \cap N(v_{j})|$ is even then it checks condition 3.

For the third condition, in order to calculate the value of $\beta_{ij}$, we will consider the
information about the overcrossing and the undercrossing strands $(O_{i},U_{i})$ for each crossing $i$ in $w$ (to easily locate the counterpart of each symbol) and a signing $w'$ of $w$.  Since we can choose any signing $w'$, we will assume that the signing $w'$ of $w$ is positive for all crossings. 
To check that two vertices $v_{i}$ and $v_{j}$ are connected by an edge of $G(w)$, 
 the automaton keeps a copy of $i$ and $j$ in the registers and then it checks that there is only one occurrence of $j$ in between $U_i$ and $O_i$.  Now to calculate the value of $\beta_{ij}$ the automaton moves its head to find the symbol $U_j$ then compares the counterpart of each symbol $k$ in between $U_{j}$ and $O_{j}$ with the symbols in the set $\bar{S}_{i} (U_{i},\ldots{},O_{i})$.If there is a match it will move first to an odd state and then alternate between odd and even states for any further matches until $O_{j}$ is reached.  Finally to traverse a closed path in $G(w)$, the automaton will choose a vertex $v_{i}$ non-deterministically, traverse along its edge and sum up the values of $\beta_{ij}$ of each visited edge by incrementing the counter by 1 (modulo 2) if the value of $\beta_{ij}$ is odd and continue updating the counter until the same vertex $v_{i}$ is met for the second time. If $v_{i}$ is met for the second time and the value of the counter is odd then the automaton moves to an accepting state.\qed

\end{proof}

\begin{corollary}
The language of planar unsigned Gauss word can be recognised by two-way Co-Non-Deterministic register automata.
\end{corollary}

\section{Definability of planarity of Gauss words}

An alternative and somewhat complementary approach to the studying descriptional complexity of 
recognisability problems is that based on definability in some logic. The question of definability of 
planarity of Gauss (multi-)words has already been addressed in the recent work by B. Courcelle \cite{courcelle-diagonal}.

Before we formulate the result from \cite{courcelle-diagonal} and discuss its relation 
 to the work we present in this paper, we would like to recall the main 
results on relationships between the computational models 
we consider in this paper and  definability in logic. 
The computational  power of classical finite 
automata over finite alphabets is characterised precisely in terms of  definability in \emph{Monadic 
Second Order Logic (MSO)}, the extension of \emph{First-Order Logic (FO)} with quantification over sets.  The classical theorem of Trakhtenbrot and Elgot at al. 
from 1950s states that the languages recognisable by finite state  automata (regular languages)  are 
exactly those definable in  MSO.  The languages definable in FO  constitute an important class 
of star-free regular languages. For the case of automata over infinite alphabets the situation is much more intricate, and in general, register 
automata are orthogonal to logically defined classes.   In \cite{neven2004fsm} authors compared 
definability in $MSO^{\ast}$, suitably defined variant  of MSO which allows to define the properties of data words,  with recognizability by register and pebble 
automata. In particular they have shown that $MSO^{\ast}$ is as least as powerful as \emph{one-way non-deterministic} register automata, 
but incomparable with \emph{two-way deterministic} and \emph{non-deterministic} automata.  Pebble automata behave much better and recognisability by 
all their natural variants is covered by definability in $MSO^{\ast}$. 

 In \cite{courcelle-diagonal} B.~Courcelle proves the following theorem, which we reformulate in the terms we have adopted in the present paper.

\begin{theorem}\cite{courcelle-diagonal}
For every genus $g$, it is definable by an MSO formula the property of a Gauss (multi-)word to be a code of the self-intersecting closing curve, 
embeddable in a surface of the genus $g$. 
\end{theorem} 

\begin{corollary}
The planarity of unsigned Gauss words is definable in MSO
\end{corollary}

\begin{note}
The encoding of Gauss (multi)words by relational structures and logic MSO used in \cite{courcelle-diagonal} are  different from the encoding of data words and logic $MSO^{\ast}$ 
from \cite{neven2004fsm}, but insignificantly. It is straightforward to show that over Gauss (multi)words the notion of definability is the same for both cases.  
\end{note}

It follows that our result on co-NRA recognisability of the planar unsigned Gauss words 
 is incomparable with the above MSO-definability result restricted to the Gauss words (as opposed to multiwords). 
On the other hand  in \cite{courcelle-diagonal} MSO-definability is shown for 
 the general case of multiwords (i.e. the case of link diagrams is covered too) and 
 for embeddability of curves in the surfaces of arbitrary genus. 
For simplicity we discussed the case of a knot, but all results clearly extend 
 to link diagrams in surfaces of any genus, because our approach is based on 
 computing the Euler characteristic of the combinatorial Carter surface that has 
 the least genus among all compact oriented surfaces containing a link diagram 
 encoded by a given collection of Gauss words \cite[Lemma~3.2]{kurlin2006gpr}.
Also of interest here would be to replace co-NRA automata with some variants of
 pebble automata and produce better upper bounds than MSO-definability. 
    
\section{Conclusion}

We have applied automata over infinite alphabets for studying 
complexity of problems related to knots. We have shown that 
the language of the signed Gauss words can be
recognised by deterministic two-way register automata, while 
for the same problem for unsigned words we demonstrated an upper 
bound in terms of  co-NRA. Our results on recognisability of planarity of Gauss words are summarised in the Table 1.
The obvious next step is to try to establish a lower and better upper bound for 
the latter problem. More generally, the automata based approach  opens 
perspectives for studying more complex knot problems, like 
unknotedness or  equivalence. Non-trivial lower bounds for such 
problems are unknown and weak automata models are plausible 
candidates here to try.   In opposite direction, knot theory 
provides a reach supply of natural problems formulated in 
terms of languages over infinite alphabets, and that, one 
may expect, will influence the development of the theory of 
such languages and related computational models.

\begin{table}[h]
\caption{Recognising planar and non-planar Gauss words by different automata models}
\centering
\begin{tabular}{|c|rr|rr|}
\hline
  & Planar &  & Non-planar&   \\
\hline

 & Signed & Unsigned & Signed & Unsigned \\
DRA& $\surd$ & ? & $\surd$ & ? \\
NRA& $\surd$& ? & $\surd$ & $\surd$ \\
Co-NRA& $\surd$&$\surd$& $\surd$&$\surd$ \\
\hline
\end{tabular}
\end{table}

In Table 1, the tick ($\surd$) indicates that the language is recognisable by 
 the corresponding automata, while the question mark (?) indicates that it is not known yet 
 whether the automaton in question can recognise the specified language.

\bibliography{knots-25}{}
\bibliographystyle{splncs}
\end{document}